\documentclass[runningheads,USenglish]{llncs}
\usepackage[T1]{fontenc}
\usepackage{graphicx}
\graphicspath{{images/}}
\usepackage[]{todonotes}
\usepackage{microtype}
\usepackage[font=small, labelfont=bf]{caption}
\usepackage[font=small, labelfont=small]{subcaption}
\usepackage{xspace}
\usepackage{amssymb}
\usepackage{xcolor}
\usepackage{amsmath}
\definecolor{defblue}{rgb}{0.121,0.47,0.705}
\definecolor{linkblue}{rgb}{0.198,0.198,0.5392}
\let\emph\relax
\DeclareTextFontCommand{\emph}{\color{defblue}\em}
\DeclareTextFontCommand{\bl}{\color{defblue}}
\usepackage[colorlinks=true, 
    linkcolor=linkblue, 
    anchorcolor=linkblue,
    citecolor=linkblue, 
    filecolor=linkblue,
    menucolor=linkblue,
    urlcolor=linkblue, 
    bookmarks=true, 
    bookmarksopen=true,
    bookmarksopenlevel=2, 
    bookmarksnumbered=true, 
    hyperindex=true,
    plainpages=false, 
    pdfpagelabels=true
]{hyperref}

\definecolor{dark red}{rgb}{0.89,0.102,0.109}
\DeclareTextFontCommand{\re}{\color{dark red}}
\definecolor{dark orange}{rgb}{1,0.498,0}
\DeclareTextFontCommand{\og}{\color{dark orange}}
\definecolor{dark green}{rgb}{0.2,0.627,0.172}
\DeclareTextFontCommand{\gr}{\color{dark green}}

\usepackage[capitalise,noabbrev,nameinlink]{cleveref}
\usepackage[inline]{enumitem}

\usepackage{thmtools,apptools,thm-restate}
\newcommand{\restateref}[1]{\IfAppendix{\hyperref[#1]{$\star$}}{\hyperref[#1*]{$\star$}}}

\renewcommand{\orcidID}[1]{\href{https://orcid.org/#1}{\includegraphics[scale=.03]{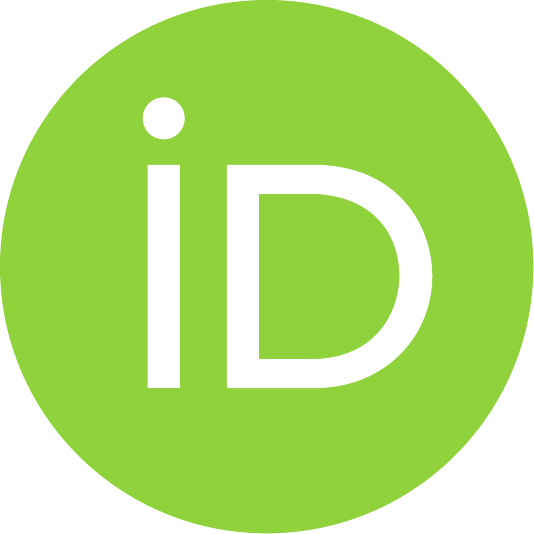}}}

\usepackage{listings} 			%
\usepackage{rotating}

\usepackage[ruled,vlined]{algorithm2e}
\SetKw{KwAnd}{and}
\SetKw{KwOr}{or}
\makeatletter
\def\BState{\State\hskip-\ALG@thistlm}
\makeatother

\let\doendproof\endproof
\renewcommand\endproof{~\hfill$\qed$\doendproof}

\usepackage{etoolbox}
\BeforeBeginEnvironment{algorithm}{\definecolor{defblue}{rgb}{0,0,0}}
\AfterEndEnvironment{algorithm}{\definecolor{defblue}{rgb}{0.121,0.47,0.705}}
\crefname{condition}{Condition}{Conditions}
\crefname{condition}{Condition}{Conditions}
\crefname{property}{Property}{Properties}
\crefname{property}{Property}{Properties}
\crefname{case}{Case}{Cases}
\crefname{case}{Case}{Cases}
\newenvironment{sketch}{\proof}{\endproof}

\begin{document}

\title{Mutual Witness Proximity Drawings of Isomorphic Trees}
\author{Carolina Haase\inst{1}\orcidID{0000-0001-6696-074X} \and
Philipp Kindermann\inst{1}\orcidID{0000-0001-5764-7719} \and
William J. Lenhart\inst{2}\orcidID{0000-0002-8618-2444} \and
Giuseppe Liotta\inst{3}\orcidID{0000-0002-2886-9694}}

\authorrunning{Haase et al.}
\institute{Universit\"at Trier, Trier, Germany\\
\email{\{haasec,kindermann\}@uni-trier.de} \and
Williams College, Williamstown, USA\\
\email{wlenhart@williams.edu} \and
Università degli Studi di Perugia, Perugia, Italy\footnote{Research partially supported by: (i) MUR PRIN Proj. 2022TS4Y3N - ``EXPAND: scalable algorithms for EXPloratory Analyses of heterogeneous and dynamic Networked Data''; (ii) MUR PRIN Proj. 2022ME9Z78 - ``NextGRAAL: Next-generation algorithms for constrained GRAph visuALization''}\\
\email{giuseppe.liotta@unipg.it}}

\maketitle

\begin{abstract}
A pair $\langle G_0, G_1 \rangle$ of graphs admits a mutual witness proximity drawing $\langle \Gamma_0, \Gamma_1 \rangle$ when: (i) $\Gamma_i$ represents $G_i$, and (ii) there is an edge $(u,v)$ in $\Gamma_i$ if and only if there is no vertex $w$ in $\Gamma_{1-i}$ that is ``too close'' to both $u$ and $v$ ($i=0,1$). In this paper, we consider infinitely many definitions of closeness by adopting the $\beta$-proximity rule for any
$\beta \in [1,\infty]$ and study pairs of isomorphic trees that admit a mutual witness $\beta$-proximity drawing. Specifically, we show that every two isomorphic trees admit a mutual witness $\beta$-proximity drawing for any $\beta \in [1,\infty]$.  The constructive technique can be made ``robust'': For some tree pairs we can suitably prune linearly many leaves from one of the two trees and still retain their mutual witness $\beta$-proximity drawability.  Notably, in the special case of isomorphic caterpillars and $\beta=1$, we construct linearly separable mutual witness Gabriel drawings. 
\end{abstract}

\section{Introduction}\label{se:intro}
Proximity drawings are geometric graphs (i.e., straight-line drawings) such that any two vertices are connected by an edge if and only if they are deemed to be close according to some definition of closeness. Therefore, proximity drawings are such that pairs of non-adjacent vertices are relatively far apart while highly connected subgraphs correspond to groups of vertices that can be naturally clustered together in a visual inspection. 

In this paper, we investigate \emph{mutual witness proximity drawings}, which employ the concept of closeness to simultaneously represent pairs of graphs. Specifically, consider a pair of graphs, denoted as $\langle G_0, G_1 \rangle$. The pair admits a mutual witness proximity drawing, denoted as $\langle \Gamma_0, \Gamma_1 \rangle$, under the following conditions: (i) $\Gamma_i$ represents $G_i$, and (ii) an edge $(u,v)$ exists in $\Gamma_i$ if and only if there is no vertex $w$ in $\Gamma_{1-i}$ that is ``too close'' to both $u$ and $v$ (where $i=0,1$). Vertex $w$ is called a \emph{witness} and its proximity to $u$ and $v$ impedes the presence of the edge. Clearly, by changing the definition of proximity a pair of graphs may or may not admit a mutual witness proximity drawing.

There is general consensus in the literature to define the closeness of $w$ to both $u$ and $v$ by means of a \emph{proximity region} of $u$ and $v$, which is a convex region in the plane whose area increases when the distance between $u$ and $v$ increases. For example, the \emph{Gabriel region}~\cite{gs-nsagv-69} of $u$ and $v$ is the disk whose diameter is the line segment $\overline{uv}$; the witness $w$ is close to $u$ and $v$ if it is a point of their Gabriel disk. A mutual witness Gabriel drawing of a pair $\langle G_0, G_1 \rangle$ is therefore a pair of drawings $\Gamma_0$ of $G_0$ and $\Gamma_1$ of $G_1$ such that for any two non-adjacent vertices in one drawing their Gabriel disk contains a witness from the other drawing, while for any two adjacent vertices their Gabriel region does not contain any witnesses. \cref{fi:intro-a} shows a mutual witness Gabriel drawing of two caterpillars. As another example, the \emph{relative neighborhood region}~\cite{DBLP:journals/pr/Toussaint80} of $u$ and $v$ is the intersection of the two disks of radius $d(u,v)$ centered at $u$ and $v$, respectively.  \cref{fi:intro-b} depicts a mutual proximity drawing that adopts the relative neighborhood region: The drawing has the same vertex set but fewer edges than the drawing in \cref{fi:intro-a}.

\begin{figure}[!t]
	\centering
		\subcaptionbox{Gabriel drawing\label{fi:intro-a}}[.38\textwidth]{\includegraphics[page=1]{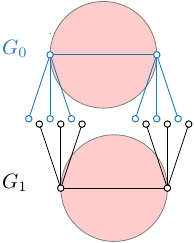}\vspace{2em}}
    \hfil
        \subcaptionbox{Mutual witness relative neighborhood drawing\label{fi:intro-b}}[.58\textwidth]{  \includegraphics[page=2]{images/caterpillar.pdf}}
	\caption{Two mutual witness drawings on the same point set.}
\end{figure}

We want to understand what families of graph pairs admit a mutual witness proximity drawing for a given definition of proximity. Intuitively, the denser the two graphs are, the more likely they admit such a representation: If the graphs are complete, we can draw them sufficiently far apart so that the proximity regions of their edges do not contain any witnesses. On the other hand, when the graphs are sparse there are many non-adjacent vertices requiring the presence of witnesses in their proximity regions, which makes the geometry of the two drawings strongly depend on one another. We specifically study very sparse graphs, namely trees. An outline of our contribution is as follows.

In Section~\ref{se:iso-caterpillars}, we prove that any pair $\langle G_0, G_1 \rangle$ of isomorphic caterpillars admits a mutual witness Gabriel drawing $\langle \Gamma_0, \Gamma_1 \rangle$ such that $\Gamma_0$ and $\Gamma_1$ are linearly separable. This is somewhat surprising as caterpillars are very sparse graphs and the linear separability of mutual witness Gabriel drawings was known only for graphs of small diameter, namely at most two~\cite{DBLP:conf/gd/LenhartL22}.

In Section~\ref{se:iso-trees}, we extend the previous result in two different directions: We consider pairs of general isomorphic trees and we study their drawability for an infinite family of proximity regions called \emph{$\beta$-regions}~\cite{KR-85}, whose shape depends on a parameter $\beta \in \mathbb{R}$. We show that any pair $\langle G_0, G_1 \rangle$ of isomorphic trees admits a mutual witness proximity drawing for any $\beta$-region such that $\beta \in [1,\infty]$. While the two drawings are no longer linearly separable, they have the property that the coordinates of their vertex sets remain the same for any possible value of $\beta$. It is worth recalling that the Gabriel disk is the $\beta$-region for $\beta=1$ and that the relative neighborhood region corresponds to the $\beta$-region for $\beta=2$.

In Section~\ref{se:non-iso-trees}, we investigate the ``robustness'' of the construction of Section~\ref{se:iso-trees}: We show that for some tree pairs, this construction can be modified so that the drawing remains valid even after pruning a suitable set of leaves. While it is known that any two star trees admit a mutual witness Gabriel drawing if and only if the cardinalities of their vertex sets differ by at most two~\cite{DBLP:conf/gd/LenhartL22}, we show that there exist tree pairs which can differ by linearly many leaves and still admit a mutual witness proximity drawing for any $\beta$-region such that $\beta \in [1,\infty]$.

Results marked with a (clickable) ``$\star$'' are proved in the appendix.

\section{Related Work}\label{se:related}

Proximity drawings are a classical research topic in graph drawing; they find application in several areas, including pattern recognition, data mining, machine learning, computational biology, and computational morphology. Proximity drawings have also been used to determine the faithfulness of large graph visualizations. A limited list of references includes~\cite{DBLP:journals/jgaa/EadesH0K17,jt-rngtr-92,DBLP:reference/crc/Liotta13,DBLP:books/wi/OkabeBSCK00,DBLP:reference/cg/ORourkeT17,DBLP:conf/mldm/ToussaintB12}.
 
In the context of designing trained classifiers, mutual witness proximity drawings were first introduced by Ichino and Slansky~\cite{DBLP:journals/pr/IchinoS85} under the name of \emph{interclass rectangle of influence graphs}. In~\cite{DBLP:journals/pr/IchinoS85} the proximity region of a pair of vertices, called the \emph{rectangle of influence}, is the smallest axis-aligned rectangle containing the two vertices. This study was then extended to other families of proximity regions, including the Gabriel region, in a sequence of papers by Aronov et al.~\cite{DBLP:journals/comgeo/AronovDH11,DBLP:journals/comgeo/AronovDH13,DBLP:journals/ipl/AronovDH14,DBLP:journals/gc/AronovDH14}. Notably, in~\cite{DBLP:journals/gc/AronovDH14} it is said that once the combinatorial properties of those pairs of graphs that admit a mutual witness Gabriel drawing are understood, \emph {``we  would have useful tools for the description of the interaction between two point sets''}. Aronov et al. prove in~\cite{DBLP:journals/ipl/AronovDH14} that any pair of complete graphs admits a mutual witness Gabriel drawing where the two drawings are linearly separable. The linear separability property of mutual witness Gabriel  drawings is extended to diameter-2 graphs by Lenhart and Liotta, who also give a complete characterization of those complete bipartite graphs that admit a mutual witness Gabriel drawing~\cite{DBLP:conf/gd/LenhartL22}. Another related contribution of Aronov et al.~\cite{DBLP:journals/comgeo/AronovDH11,DBLP:journals/comgeo/AronovDH13,DBLP:journals/ipl/AronovDH14,DBLP:journals/gc/AronovDH14} is to introduce and study \emph{witness proximity drawings}, which can be shortly described as a relaxation of mutual proximity drawings where one of the two drawings has no edges, independently of whether the proximity regions of its vertices do or do not contain any witnesses. 

\section{Preliminaries}\label{se:preli}

We assume familiarity with basic graph drawing concepts; see e.g.~\cite{DBLP:books/ph/BattistaETT99,DBLP:conf/dagstuhl/1999dg,DBLP:reference/crc/2013gd,DBLP:reference/cg/TamassiaL04}.

 Let $p$ and $q$ be two distinct points in the plane. We denote by $\overline{pq}$ the straight-line segment having $p$ and $q$ as its extreme points. 
We define $\beta$-regions adopting the notation in~\cite{DBLP:journals/jda/BattistaLW06}. A region in the plane is \emph{open} if it is an open set, that is the points on its boundary are not part of the region, and \emph{closed} if all of the points of the boundary are part of the region.
Given a pair $p,q$ of points in the plane and a real number $\beta \in [1,\infty]$, the \emph{open $\beta$-region} of $p$ and $q$, denoted by $R(p,q,\beta)$, is defined as follows. 
For $1 \leq \beta < \infty$, $R(p,q,\beta)$ is the intersection
of the two open disks of radius $\beta d(p,q)/2$ and centered at the points
$(1 - \beta/2)p + (\beta/2)q$ and $(\beta/2)p + (1 - \beta/2)q$.
$R(p,q,\infty )$ is the open infinite strip perpendicular to the line segment $\overline{pq}$ and for $\beta \in [1,\infty]$, the closed $\beta$-region $R[p,q,\beta]$ is simply the open region $R(p,q,\beta)$ along with its boundary; see \cref{fi:beta-regions}.
\begin{figure}[t]
	\centering
	\includegraphics{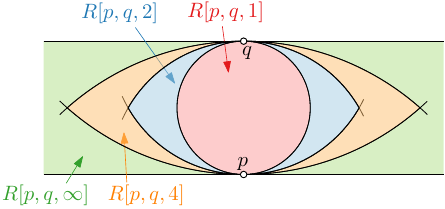}
	\caption{Examples of $\beta$-proximity regions for $\beta \geq 1$.}
	\label{fi:beta-regions}
\end{figure}

Note that $R[p,q,1]$ is the Gabriel region of $p,q$ and that $R(p,q,2)$ is the relative neighborhood region of $p,q$. We shall denote as a  \emph{MW-$[\beta]$ drawing} a mutual witness proximity drawing such that for any two vertices $p$ and $q$ the proximity region is  $R[p,q,\beta]$. In particular, an \emph{MW-$[1]$ drawing} is a mutual witness Gabriel drawing. Similarly, a \emph{MW-$(\beta)$ drawing} is a mutual witness proximity drawing that uses the open $\beta$-region.

As we shall see, some of our constructive arguments produce drawings that are simultaneously \emph{MW-$(\beta)$} and \emph{MW-$[\beta]$} drawings; in this case we refer to them simply as \emph{MW-$\beta$ drawings}. Note that for any pair $p,q$ of vertices in an MW-$\beta$ drawing, $R(p,q,\beta)$ contains a witness if $p$ and $q$ are not adjacent, while $R[p,q,\beta]$ contains no witnesses if $p$ and $q$ are adjacent.

Let $\langle \Gamma_0, \Gamma_1 \rangle$ be an MW-$\beta$ drawing of graphs $\langle G_0,G_1\rangle$ for some value of $\beta$. We say that the drawing is \emph{linearly separable} if there exists a line $\ell$ such that $\Gamma_0$ and $\Gamma_1$ lie in opposite half-planes with respect to $\ell$. 
The following property rephrases an observation of~\cite{DBLP:conf/gd/LenhartL22} and will be used in the proof of \cref{th:caterpillar}.
\begin{property} \label{pr:witnessInStrip}
    Let $\langle \Gamma_0, \Gamma_1 \rangle$ be a linearly separable MW-$[1]$ drawing and let $u$ and $v$ be any two non-adjacent vertices of $\Gamma_i$, for $i=0,1$. Then any witness in  $R[u,v,1]$ is also a point in $R[u,v,\infty]$ 
\end{property}

\section{MW-$[1]$ Drawings of Isomorphic Caterpillars}\label{se:iso-caterpillars}
\begin{figure}[!t]
	\centering
	\includegraphics{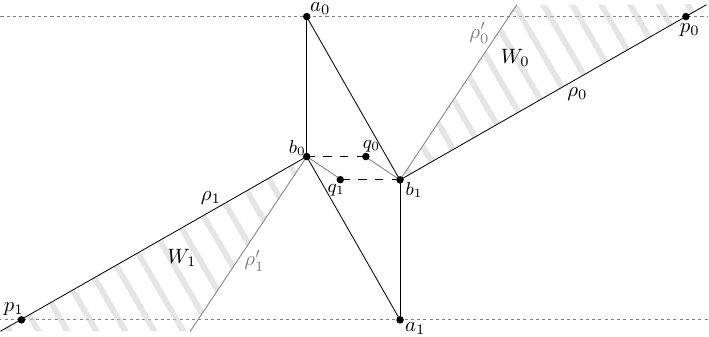}
	\caption{A winged~parallelogram~with~anchors~$q_0,q_1$,~safe~wedges~$W_0,W_1$,~and ports~$p_0,p_1$.}
	\label{fi:winged}
\end{figure} 
A \emph{caterpillar} is a tree $T$ such that, when removing the leaves of $T$ one is left with a non-empty path called \emph{spine} of $T$. We call the graph $K_{1,n}$ with $n \geq 1$ a \emph{star}; if $n>1$ the non-leaf vertex of a star is the \emph{center} of the star, otherwise (i.e., when the star is an edge) either vertex can be chosen as the center.

In this section, we prove that any two isomorphic caterpillars admit a linearly separable MW-[$1$] drawing, that is they admit a linearly separable mutual witness Gabriel drawing. As pointed out both in~\cite{DBLP:journals/ipl/AronovDH14} and in~\cite{DBLP:journals/pr/IchinoS85}, the linear separability of mutual witness proximity drawings is a desirable property because it gives useful information about the inter-class structure of two sets of points.

Let $P=\langle a_0,b_0,a_1,b_1\rangle$ be a parallelogram such that 
$y(a_0) > y(b_0) > y(b_1) > y(a_1)$ and $x(a_0) = x(b_0) < x(a_1) = x(b_1)$. Let $q_0$ and $q_1$ be two points in the interior of $P$ satisfying $y(b_i) = y(q_i)$, $x(q_1) < x(q_0)$, and $x(q_0) - x(b_0) = x(b_1) - x(q_1)$. Let $W_i$ be the wedge with apex $b_i$ not containing any vertex of $P$ other than $b_i$ and defined by two rays $\rho_i,\rho'_i$ such that $\rho_i$ is perpendicular to  $\overline{a_i b_{1-i}}$ and $\rho'_i$ is perpendicular to  $\overline{q_i b_{1-i}}$. We call $W_i$ \emph{safe wedges} of $P$ and the $q_i$  \emph{anchors}. We assume $W_i$ to be an open set. Finally, we identify two \emph{ports}, the points $p_i$, where $p_i$ is the point along $\rho_i$ such that $y(p_i) = y(a_i)$. The parallelogram $P$ together with its anchors, safe wedges, and ports is called a \emph{winged parallelogram} $WP(P, q_0, q_1, W_0, W_1, p_0, p_1)$. \cref{fi:winged} shows an example of a winged parallelogram. The following property is an immediate consequence of the definition of winged parallelogram; see also \cref{fi:property}. 

\begin{property}\label{pr:winged}
    Let $WP(P, q_0, q_1, W_0, W_1, p_0, p_1)$ be a winged parallelogram such that the interior angles at points $a_i$ ($i=0,1$) are at most $\frac{\pi}{4}$. Let $s_i, t_i, z_i$ ($i=0,1$) be any three points such that $s_i \in \overline{b_i q_i}$, $t_i \notin W_i$ with $x(t_0)\ge x(p_0)$, $x(t_1)\le x(p_1)$, and $z_i\in W_i$ with $y(z_i)=y(a_i)$. Then: 
    \begin{enumerate*}[label=(P\arabic*)]
        \item[\addtocounter{enumi}{1}\gr{\labelenumi}]\label[property]{prop:saz} neither $s_{1-i}$ nor $a_{1-i}$ are points of $R[a_i,z_i,1]$; 
        \item[\addtocounter{enumi}{1}\bl{\theenumi}]\label[property]{prop:bst} $b_{1-i} \in R(s_i,t_i,1)$; 
        \item[\addtocounter{enumi}{1}\re{\theenumi}]\label[property]{prop:bat} $b_{1-i} \in R[a_i,t_i,1]$ if $t_i$ on $\rho_i$ and $b_{1-i} \in R(a_i,t_i,1)$ if $t_i$ is not on $\rho_i$.
    \end{enumerate*}
\end{property}

\begin{figure}[!t]
	\centering
	\includegraphics{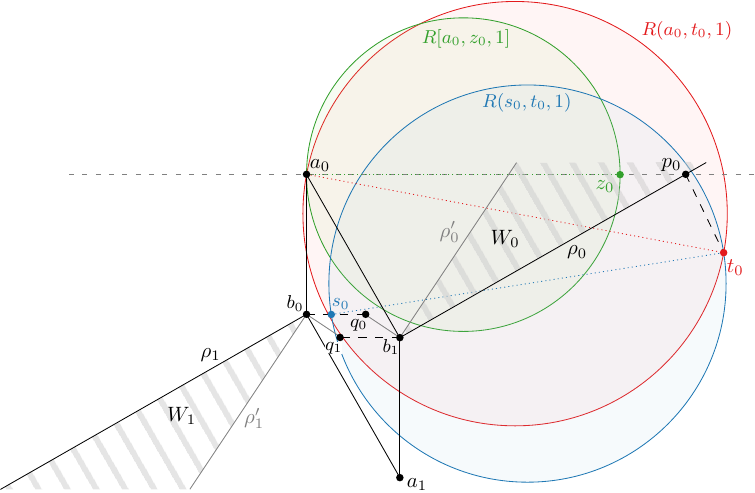}
	\caption{Illustration for \cref{pr:winged}. 
            \gr{(P1)} Neither $a_1$ nor any points of $\overline{q_1 b_{1}}$ are points of $R[a_0,z_0,1]$; 
            \bl{(P2)} $b_1 \in R(s_0,t_0,1)$;
            \re{(P3)} $b_1 \in R(a_0,t_0,1)$.}
	\label{fi:property}
\end{figure}
 We first show how to draw pairs of isomorphic stars into a winged parallelogram and then generalize the construction to pairs of isomorphic caterpillars.   
 
\begin{restatable}[\restateref{le:stars}]{lemma}{LemStars}
\label{le:stars}
Let $\langle T_0,T_1 \rangle$ be a pair of isomorphic stars such that, for $i=0,1$, $T_i$ has root $r_i$ and leaves $v_{i,0}, \ldots, v_{i,k}$. Then $\langle T_0,T_1 \rangle$ admits an MW-$[1]$ drawing $\langle \Gamma_0,\Gamma_1 \rangle$  contained in a winged parallelogram $WP(P, q_0, q_1, W_0, W_1,p_0,p_1)$ such that: (i) $r_i$ is drawn at $a_i$ and the the internal angle of $WP(P, q_0, q_1, W_0, W_1,p_0,p_1)$ at $a_i$ is at most $\frac{\pi}{4}$; (ii) $v_{i,0}$ is drawn at $b_i$; and (iii) for $0 < j \leq k$, $v_{i,j}$ is drawn at an interior point of the segment $\overline{b_i q_i}$.
\end{restatable}
\begin{sketch}
For $i=0,1$, if $T_i$ has only one leaf, the construction is trivial; see \cref{fi:stars-zero}.
Otherwise, we draw the leaves of $T_i$ uniformly spaced along a horizontal segment $\sigma_i$ 
and then place $\sigma_0$ and $\sigma_1$ relative to each other so that for every pair of consecutive leaves of $T_i$, there is a witness for that pair among the leaves of $T_{1-i}$; see \cref{fi:stars-leaves}.

 The horizontal line midway between $\sigma_0$ and $\sigma_1$ will form a separating line for $\langle \Gamma_0, \Gamma_1 \rangle$ once the centers $r_i$ of $T_i$ are placed.
 The center $r_0$ of $T_0$ is then placed vertically above the leftmost leaf of $T_0$ and the center $r_1$ of $T_1$ is placed vertically below the rightmost leaf of $T_1$, each center far enough from the separating line so that for $i=0,1$ and $0\leq j \leq k$, no proximity region $R[r_i,v_{i,j},1]$ contains any witness from $T_{1-i}$; see \cref{fi:stars-root}.
\end{sketch}

\begin{figure}[!t]
	\centering
 
    \begin{minipage}[b][5.3cm]{.4\columnwidth}
	 \subcaptionbox{Drawing for $k=0$\label{fi:stars-zero}}[\linewidth]{\includegraphics[page=5]{images/stars-construction.pdf}}
      \vfill
	 \subcaptionbox{Placement of the leaves\label{fi:stars-leaves}}{\includegraphics[page=1]{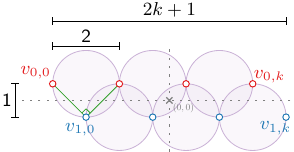}}
    \end{minipage}
    \hfill
    \subcaptionbox{Placement of the root\\ (distorted for readability)\label{fi:stars-root}}{~\includegraphics[page=2]{images/stars-construction.pdf}~}
    \hfill
    \subcaptionbox{Drawing for $k{=}3$\label{fi:stars-drawing}}{\includegraphics[page=4]{images/stars-construction.pdf}}
	\caption{Illustration for the Proof of \cref{le:stars}.}
	\label{fi:stars-construction}
\end{figure}

In the following we call an MW-$[1]$ drawing $\langle \Gamma_0,\Gamma_1 \rangle$  of two isomorphic stars computed as in the proof of \cref{le:stars} a \emph{WP-drawing on $P$} and say that the winged parallelogram \emph{supports} the drawing; see \cref{fi:stars}. Note that, by construction, the horizontal line $L$ having $y(L) = (y(b_0) + y(b_1))/2$ is a separating line for the WP-drawing of two isomorphic stars. 

\begin{figure}[t]
	\centering
	\includegraphics{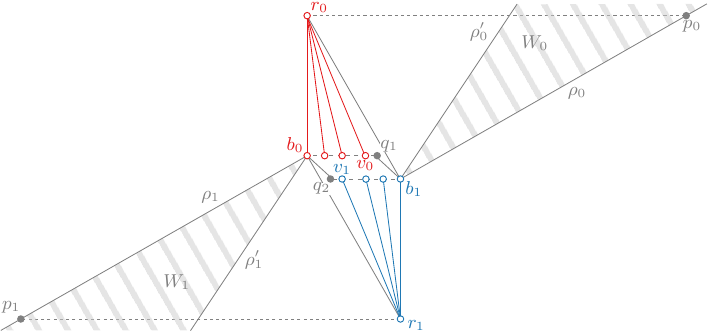}
	\caption{A WP-drawing of two isomorphic stars on a parallelogram $P$}
	\label{fi:stars}
\end{figure}

\begin{lemma}\label{le:smaller-stars}
    Let $\langle \Gamma_0,\Gamma_1 \rangle$ be a WP-drawing of two isomorphic stars $\langle T_0,T_1 \rangle$ and let $P$ be the winged parallelogram that supports $\langle \Gamma_0,\Gamma_1 \rangle$. Then, any pair $\langle T'_0,T'_1 \rangle$ of isomorphic stars with at least one leaf and $T'_i \subset T_i$ has a WP-drawing on~$P$.
\end{lemma}
\begin{proof}
    Let $r_i$ be the root of $T_i$ and $v_{i,0},\ldots,v_{i,k}$ be the leaves of $T_i$. Consider the drawing $\langle \Gamma_0,\Gamma_1\rangle$ computed in \cref{le:stars}; see \cref{fi:stars-construction}. We use the same notation as in the proof of \cref{le:stars}.
    Remove all leaves $v_{i,j}$ that are not in $T_i'$ and reposition the remaining leaves uniformly along $\sigma_i$ as in the proof of \cref{le:stars}.
        
    By construction, the Gabriel region $R[v_{0,i},v_{0,j},1]$ for every $v_{0,i},v_{0,j}\in T_i, 1\le i<j\le k$ still contains the vertex $v_{1,i}$, while the Gabriel region $R[v_{1,i},v_{1,j},1]$ for every $v_{1,i},v_{1,j}\in T_i, 1\le i<j\le k$ still contains the vertex $v_{1,j}$. 
    Otherwise, if $v_{i,0}\notin T'_i$, then take any leaf $v_{i,j}\in T'_i$, switch its position with $v_{i,0}$ in $\Gamma_i$, and then proceed as above.
\end{proof}

\begin{theorem}\label{th:caterpillar}
Any pair $\langle T_0,T_1\rangle$ of isomorphic caterpillars admits a linearly separable MW-$[1]$ drawing.
\end{theorem}

\begin{proof}
For $i = 0,1$, if each $T_i$ is a path, the pair can easily be realized by two horizontal paths, such that corresponding vertices of $T_0$ and $T_1$ have the same $x$-coordinates, all edges have the same length and the $y$-distance between $T_0$ and $T_1$ is at most the edge length. So we can assume that the spine of $T_i$ is a path such that at least one spine vertex has degree greater than two.

Let $r_{i,0}, \dots, r_{i,k}$ be the spine vertices of $T_i$ in the order that they appear along the spine.  Decompose $T_i$ into subtrees $T_{i,0}, \ldots T_{i,k}$, having roots $r_{i,0}, \ldots, r_{i,k}$ respectively. Note that each $\langle T_{0,j},T_{1,j} \rangle$ is either an isomorphic pair of stars with centers $r_{0,j}$ and $r_{1,j}$, respectively, or it is a pair of isolated vertices.

Let $h$ be an index such that $r_{i,h}$ is a vertex of highest degree in $T_i$. Compute a WP-drawing $\langle \Gamma_{0,h},\Gamma_{1,h} \rangle$ of  $\langle T_{0,h},T_{1,h} \rangle$ by means of \cref{le:stars} and let $WP_h = WP(P_h, q_{0,h}, q_{1,h}, W_{0,h}, W_{1,h}, p_{0,h}, p_{1,h})$ be the winged parallelogram that supports the drawing.
Let $N=y(r_{0,h})$ and $S= y(r_{1,h})$ and let $L_0$ and $L_1$ be the two horizontal lines at heights $N$ and $S$, respectively. We will construct a  MW-$[1]$ drawing of the two caterpillars such that all spine vertices of $T_i$ lie on $L_i$ and such that the horizontal line $L$ at height $(N+S)/2$ separates $T_0$ from $T_1$.

For any $0 \leq j \leq k$, $j\neq h$ such that $r_{i,j}$ has at least one leaf, we use \cref{le:smaller-stars} to compute a WP-drawing $\langle \Gamma_{0,j},\Gamma_{1,j} \rangle$ of $\langle T_{0,j},T_{1,j} \rangle$ in a winged parallelogram $WP_j$ congruent to $WP_h$ that will be placed so that $r_{i,j}$ lies on $L_i$.
For any $0 \leq j \leq k$, $j\neq h$ such that  $r_{i,j}$ has no children, we will place $r_{i,j}$ on $L_i$
so that the line through $r_{0,j}, r_{1,j}$ is perpendicular to the line through $r_{1,h}, p_{0,h}$. 

We now describe how to place each pair $\langle \Gamma_{0,j}, \Gamma_{1,j} \rangle$; note that placing $r_{0,j}$ completely determines the placement of $\langle \Gamma_{0,j}, \Gamma_{1,j} \rangle$.
Vertex $r_{0,0}$ can be placed arbitrarily along $L_0$.
Assume now that, for some $j \geq 0$, the pairs
$\lbrace r_{0,0}, r_{1,0} \rbrace, \ldots
\lbrace r_{0,j}, r_{1,j} \rbrace$ have been placed along $L_0$ and $L_1$. We describe how to place $r_{0,j+1}$. There are three cases; see \cref{fi:lin-sep-caterpillar}:
\begin{enumerate*}[label=(\arabic*)]
\item\label[case]{case:caterpillar-1leaf} If $r_{0,j}$ has at least one leaf, place $r_{0,j+1}$ at port $p_{0,j}$.
\item\label[case]{case:caterpillar-01leaf} If $r_{0,j}$ has no leaves and $r_{0,j+1}$ has at least one leaf, place $r_{0,j+1}$ so that
$r_{1,j}$ is at port $p_{1,j+1}$.
\item\label[case]{case:caterpillar-00leaf} If both $r_{0,j}$ and $r_{0,j+1}$ have no leaves, place $r_{0,j+1}$ at the intersection of $L_0$ with the line through $r_{1,j}$ that is perpendicular to  $\overline{r_{0,j} r_{1,j}}$.
\end{enumerate*}

\begin{figure}[!t]
	\centering
	\subcaptionbox{\cref{case:caterpillar-1leaf}\label{fi:lin-sep-caterpillar-case1}}{\includegraphics[page=1]{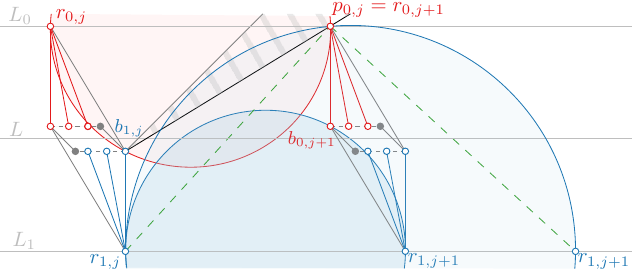}}

    \medskip
 
	\subcaptionbox{\cref{case:caterpillar-01leaf,case:caterpillar-00leaf}\label{fi:lin-sep-caterpillar-case2}}{\includegraphics[page=2]{images/lin-sep-caterpillar.pdf}}
	\caption{Illustration for the Proof of \cref{th:caterpillar}.}
    \label{fi:lin-sep-caterpillar}
\end{figure}

This construction is \textit{almost} an MW-$[1]$ drawing of $\langle T_0, T_1 \rangle$. Consider the mutual witness Gabriel drawing $\Gamma$ induced by the placement of the vertices of $\langle T_0, T_1 \rangle$ described above.
Note that in our constructed drawing:
\begin{enumerate*}[label=(\roman*)]
\item The pairs $\langle T_{0,j}, T_{1,j} \rangle$
are drawn in vertically disjoint strips and by \cref{pr:witnessInStrip} form MW-$[1]$ drawings of those pairs. 

\item For any non-spine vertex $u_{0,j} \in T_{0,j}$, and \textit{any} vertex $u_{0,t} \in T_{0,t}$ ($0\leq j < t \leq k$), by \cref{prop:bst}~{\bl (P2)}, $b_{1,j} \in R(u_{0,j},u_{0,t},1)$ and so the pair $\lbrace u_{0,j}, u_{0,t} \rbrace$ is not an edge in $\Gamma$. 
\item For any spine vertex $r_{0,j} \in T_{0,j}$, and non-spine vertex $u_{0,t} \in T_{0,t}$ ($0 \leq j < t \leq k$), either $r_{0,j}$ has a leaf, and so by \cref{prop:bat}~{\re (P3)}, $b_{1,j} \in R(r_{0,j},u_{0,t},1)$ or
$r_{0,j}$ has no leaves and $r_{1,j}\in R[r_{0,j},u_{0,t},1]$ by the construction described above.
\end{enumerate*}
Similar statements hold for pairs of vertices in $T_1$ by the symmetry of the construction.

The drawing $\Gamma$ is not yet an MW-$[1]$ drawing  of $\langle T_0, T_1 \rangle$ because %
there are \textit{no} edges in $\Gamma$ between \textit{any} pair of consecutive spine vertices of $T_i$. This problem can be easily rectified, however. Note that in $\Gamma$ there are only two types of non-adjacent vertex pairs that only have witnesses on the boundaries of their Gabriel regions (that is, that only have witnesses forming right angles), namely,  consecutive leaves in an individual subtree $T_{i,j}$,  and consecutive spine vertices  in $T_i$.
Let $r_{i,j}$ and $r_{i,j+1}$ be any two consecutive spine vertices of $T_i$. We can always perturb $\Gamma$ so that by very slightly moving to the left all vertices of $\langle T_{0,j+1}, T_{1,j+1} \rangle$, we have, by \cref{prop:saz}~{\gr (P1)}, that $R[r_{i,j},r_{i,j+1},1]$  contains no witnesses while for every other pair of non-adjacent vertices their Gabriel regions still contain a witness. Once all spine vertices have been properly connected, the resulting drawing is a linearly separable MW-$[1]$ drawing of $\langle T_0, T_1 \rangle$.
\end{proof}

\section{MW-$\beta$ Drawings of Isomorphic Trees}\label{se:iso-trees}

In this section we show that, at the expense of losing linear separability, the result of~\cref{th:caterpillar} can be extended to any two isomorphic trees and to any mutual witness proximity drawing that adopts either the open or the closed $\beta$-region for all values of $\beta \in [1,\infty]$. A nice property of our algorithm is that it does not depend on the exact choice of $\beta$, i.e., it produces a single drawing that is an MW-$\beta$ proximity drawing for every $\beta\ge 1$.

Similar to the previous section, we show a construction to recursively draw subtrees inside suitable parallelograms, which are however not winged parallelograms.  We start by defining these parallelograms. In the remainder of the section, we shall sometimes assume that our trees are rooted, in which case we denote as $(T,r)$ a tree $T$ with root $r$.

Let $P = \langle a_0,b_0,a_1,b_1 \rangle$ be a parallelogram where $\overline{a_0a_1}$ is the longer diagonal and no angle is equal to $\frac{\pi}{2}$. We say that $P$ is \emph{nicely oriented} if $y(a_0) > y(b_1) > y(b_0) > y(a_1)$ and $x(a_0) < x(b_0) < x(b_1) < x(a_1)$; see \cref{fig:nicely-oriented}. 

Let $\langle (T_0, r_0), (T_1, r_1) \rangle$ be a pair of isomorphic rooted trees with $n$ vertices each. An \emph{MW-$\beta$ parallelogram drawing} of $\langle (T_0, r_0), (T_1, r_1)\rangle$ is an MW-$\beta$ proximity drawing $\langle \Gamma_0, \Gamma_1\rangle$ contained in a nicely oriented parallelogram $P = \langle a_0,b_0,a_1,b_1 \rangle$ such that, for $i=0,1$, the following holds: 
\begin{enumerate*}[label=(\roman*)]
	\item point $a_i$ represents the root $r_i$ of $T_i$; 
    \item if $n>0$, point $b_i$ represents a vertex of $T_i$ adjacent to $r_i$;
    \item for every other vertex $v_i \in \Gamma_i$ such that $v_i$ is neither the root of $T_i$ nor the vertex at $b_i$, we have $y(b_1) > y(v_i) > y(b_0)$;
    \item no edge of $\Gamma_i$ is a vertical segment.
\end{enumerate*}
\cref{fig:parallelogram-drawing} shows an example of an MW-1 parallelogram drawing.
 
\begin{figure}[t]
    \centering
    \subcaptionbox{\label{fig:nicely-oriented}}{\includegraphics[page=1]{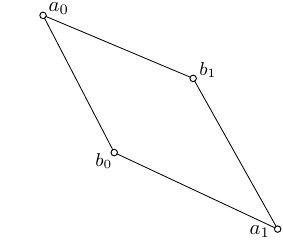}}
    \hfil
    \subcaptionbox{\label{fig:parallelogram-drawing}}{\includegraphics[page=2]{images/nicely-oriented.pdf}}
    \caption{(\subref{fig:nicely-oriented}) A nicely oriented parallelogram; (\subref{fig:parallelogram-drawing}) an MW-$1$ parallelogram drawing.}
\end{figure}

\begin{theorem}\label{th:main}
Any two isomorphic trees $\langle T_0, T_1 \rangle$ admit a parallelogram drawing that is an MW-$\beta$-drawing for all $\beta \in [1, \infty]$.
\end{theorem}

\begin{proof}
Let $r_0$ be any vertex of $T_0$ and let $r_1 \in T_1$ be the isomorphic image of $r_0$.
We will show by induction on the depth $\delta$ of $(T_0, r_0)$ that $\langle (T_0,r_0), (T_1,r_1) \rangle$ admits an MW-$\beta$ parallelogram drawing of $\langle T_0, T_1 \rangle$ for any $\beta \geq 1$.

If $\delta = 0$ each $T_i$ consists of only its root $r_i$. Choosing any nicely oriented parallelogram with $\overline{r_0r_1}$ as its long diagonal will result in a valid MW-$\beta$ drawing.
Assume the claim holds for $\delta \leq k$ and suppose $\delta = k+1$.

Let $\langle (T_{0,0}, r_{0,1}), (T_{1,0}, r_{1,0}) \rangle, \dots, \langle (T_{0,m}, r_{0,m}), (T_{1,m}, r_{1,m}) \rangle$ be the pairs of isomorphic rooted trees resulting from deleting $r_i$ from $T_i$.
By induction, each $\langle (T_{0,j}, r_{0,j}), (T_{1,j}, r_{1,j}) \rangle$ with $0 \leq j \leq m$ admits a parallelogram drawing which is an MW-$\beta$ drawing.
Let $H$ be any horizontal strip defined by two parallel lines $y = s$ and $y=t$ such that $s < t$.
We uniformly scale and translate the parallelogram drawings of $\langle (T_{0,j}, r_{0,j}), (T_{1,j}, r_{1,j}) \rangle$ such that $y(r_{0,j}) = t$ and $y(r_{1,j}) = s$. Note that this operation does not change any of the $\beta$-proximity properties of any of the tree pairs.

Let $P_j=(a_{0,j},b_{0,j},a_{1,j},b_{1,j})$ be the parallelogram that supports the MW-$\beta$ drawing $\langle (\Gamma_{0,j}), (\Gamma_{1,j}) \rangle$ of $\langle (T_{0,j}, r_{0,j}), (T_{1,j}, r_{1,j}) \rangle$. Let $\ell_j$ and $\ell'_j$ be two half-lines such that $\ell_j$ starts at $r_{0,j}$, is orthogonal to $\overline{r_{0,j}, b_{0,j}}$, and crosses $H$, and $\ell'_j$ starts at $r_{1,j}$, is orthogonal to $\overline{r_{1,j}, b_{1,j}}$, and crosses $H$; see \cref{img:subtrees}.
We position $P_{j+1}$ such that 
\begin{enumerate*}[label=(\roman*)]
	\item\label[condition]{cond1} $\ell_{j+1}$ is to the right of $\ell'_j$;
	\item\label[condition]{cond2} for any edge $e_{1,j} = (u_{1,j}, v_{1,j})$ in $T_{1,j}$, $r_{0,j+1}$ is to the right of the rightmost intersection point between $H$ and $R[u_{1,j}, v_{1,j},\infty]$ (since by inductive hypothesis no edge of $\Gamma_{1,j}$ is vertical, the coordinates of such points are finite); and
    \item\label[condition]{cond3} for any edge $e_{0,j+1} = (u_{0,j+1}, v_{0,j+1}) \in T_{0,j+1}$, $r_{1,j}$ is to the left of the leftmost intersection point between $H$ and  $R[u_{0,j+1}, v_{0,j+1},\infty]$ (by inductive hypothesis, the coordinates of such points are finite).
\end{enumerate*}

\cref{cond1} guarantees that for any vertices $v_{1,j+1} \in \Gamma_{1,j+1}$ and $v_{1,j} \in \Gamma_{1,j}$, we have $\angle(v_{1,j+1},r_{0,j+1}, v_{1,j}) > \frac{\pi}{2}$ and thus $r_{0,j+1} \in R(v_{1,j+1},v_{1,j},1)$ and $r_{0,j+1} \in R(v_{1,j+1},v_{1,j},\beta)$ for any $\beta \geq 1$. Similarly, for any vertices $v_{0,j+1} \in \Gamma_{0,j+1}$ and $v_{0,j} \in \Gamma_{0,j}$, we have $r_{1,j}\in R(v_{0,j+1},v_{0,j},\beta)$ for any $\beta \geq 1$. \cref{cond2,cond3} guarantee that for any pair of adjacent $v_{i,j},u_{i,j}$ in $\Gamma_{i,j}$, there is no witness in $R[v_{i,j},u_{i,j},\infty]$ and thus no witness in $R[v_{i,j},u_{i,j},\beta]$ for any finite $\beta \geq 1$.
\begin{figure}[t]
	\centering
	\includegraphics{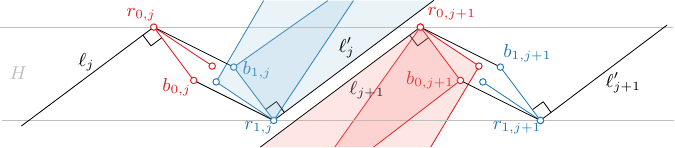}
	\caption{Parallelograms $P_j$ and $P_{j+1}$ placed inside $H$.}
	\label{img:subtrees}
\end{figure}
We now show how to place the roots $r_0 \in T_0$ and $r_1 \in T_1$ to produce an MW-$\beta$ parallelogram drawing of $\langle (T_0, r_0), (T_1, r_1) \rangle$ for any $\beta \in[1,\infty]$.

Let $L_0$ be the vertical line through $r_{0,0}$ and let $L_1$ be the vertical line through $r_{1,m}$; see \cref{fi:tree1}. We show how to place $r_i$ on $L_i$ such that the \emph{closed} $\beta$-region $R[r_i,r_{i,j},\infty]$ does not contain any witness, while for any other vertex $v \in T_i$, the \emph{open} $\beta$-region $R(r_i,v,1)$ contains a witness. This implies that $R[r_i,r_{i,j},\beta]$ does not contain any witnesses for all finite values of $\beta$ and that $R(r_i,v,\beta)$ contains some witnesses for every $\beta \geq 1$.

\begin{figure}[t]
 \centering
		\subcaptionbox{Construction of $L_0,L_1,h_0,h_1,\alpha_1,f_0$\label{fi:tree1}}{\includegraphics[page=1]{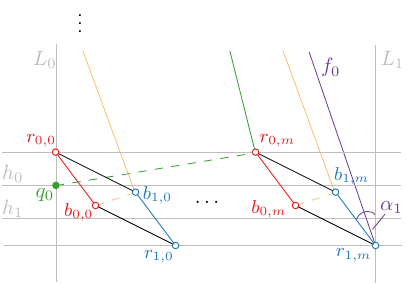}}
    \hfil
        \subcaptionbox{Construction of $I'_0,I''_0,I'''_0$\label{fi:tree2}}[.35\textwidth]{  \includegraphics[page=2]{images/trees2.pdf}}  
	\caption{Placing $r_0$ in the proof of \cref{th:main}.}
  \label{fi:root-placement}  
\end{figure}

We proceed in three steps. In the first step, we identify an interval $I'_i$ of $L_i$ such that for any point $p'\in I'_i$ and for any $r_{i,j}$, $R[p',r_{i,j},\beta]$ contains no witnesses ($i=0,1$, $0 \leq j \leq m$).
In the second step, we identify an interval $I_i''$ of $L_i$ such that for each vertex $v \neq r_{i,j}$ ($i=0,1$, $0 \leq j \leq m$) in $\Gamma_i$ and for each point $p'' \in I''_i$,   $R(p'',v,\beta)$ contains a witness.
In the third step, we identify an interval $I'''_i$ of $L_i$ such that for any point $p_0 \in I_0'''$ the segment $\overline{p_0b_{1,m}}$ does not intersect any parallelogram $P_j$ with $0 \leq j \leq m$. Similarly, for any point $p_1 \in I_1'''$, the segment $\overline{p_1,b_{0,0}}$ does not intersect $P_j$ for $0 \leq j \leq m$. As we will see, $I'_i \cap I_i'' \cap I_i'''$ is a half-infinite strip for $i = 0,1$; see \cref{fi:tree2}. We will describe how to obtain the intervals $I'_0,I''_0,I'''_0$; the intervals $I'_1,I''_1,I'''_1$ can be constructed symmetrically.

We start by defining $I_0'$. By construction of the MW-$\beta$ drawing of the forests $T_{0,0}, \dots T_{0,m}$ and $T_{1,0} \dots T_{1,m}$, there exist horizontal lines $h_0, h_1$ in the interior of $H$ such that $h_i$ separates $r_{i,0}, \dots r_{i,m}$ from every other vertex in the forest; see \cref{fi:tree1}. Let $q_0$ be the intersection point of $h_0$ and $L_0$.
Let $z'_0$ be the intersection point of $L_0$ with the line through $r_{0,m}$ perpendicular to $\overline{r_{0,m}q_0}$.
Let $I'_0 = \lbrace z \in L_0 : y(z) \geq y(z'_0) \rbrace$. 
Observe that for any $p_0 \in I'_0$ and any $r_{0,j}$, $R[p_0,r_{0,j},\infty]$ contains no witnesses.

We now define $I_0''$. For any parallelogram $P_j$ and any vertex $v_{0,j} \in T_{0,j}\setminus\{r_{0,j}\}$, 
let $z_{0,j}$ be the intersection of $L_0$ with the line through $b_{1,j}$ perpendicular to $\overline{b_{1,j}v_{0,j}}$. Let $z_{0}''$ be the $z_{0,j}$ of maximum $y$-value over all $z_{0,j}$ ($0 \leq j \leq m$). Let $I''_0 = \lbrace z \in L_0 : y(z) \geq y(z_{0}'') \rbrace$. 
Observe that for any point $p_0 \in I_0''$ and for any $v_{0,j} \in T_{0,j}\setminus\{r_{0,j}\}$, we have that $\angle(v_{0,j}, b_{1,j}, p_0) \geq \frac{\pi}{2}$ and thus $b_{1,j} \in R[p_0,v_{0,j},1]$.

We now define $I_0'''$.
Let $\alpha_0$ be the acute angle formed by $L_0$ and the segment $\overline{r_{0,0}b_{0,0}}$. Let $\alpha_1$ be the acute angle formed by $L_1$ and the segment $\overline{r_{1,m}b_{1,m}}$ and let $\alpha = \min\{\alpha_0, \alpha_1\}$. 
Let $f_0$ be a half-line starting at $r_{1,m}$, having negative slope, and forming an acute angle of $\alpha/2$ with $L_1$.
Let $z_0'''$ be $f_0 \cap L_0$ and let $I'''_0 = \lbrace z \in L_0 : y(z) \geq y(z'''_0) \rbrace$.

Let $I_i = I_i' \cap I_i'' \cap I_i'''$ and 
let $p_i \in I_i$ be such that $\overline{p_0r_{1,m}}$ is parallel to $\overline{p_1r_{0,0}}$. We draw $r_i$ at $p_i$, which produces an MW-$\beta$ drawing of $\langle T_0,T_1 \rangle$ in a parallelogram $P = \langle a_0,b_0,a_1,b_1 \rangle = \langle r_0,b_{1,m},r_1,b_{0,0} \rangle$. 
This is however not yet a parallelogram drawing, as $y(b_0) = y(b_{0,0}) > y(b_{1,m}) = y(b_1)$ and some edges are vertical. 

\begin{figure}[t]
    \centering
    \subcaptionbox{Definition of $\gamma_0,\gamma_1,\hat f_0,\hat f_1$\label{fig:rotation-before}}{\includegraphics[page=1]{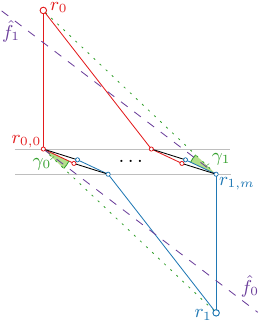}}
    \hfill
    \subcaptionbox{After rotation\label{fig:rotation-after}}{\includegraphics[page=2]{images/rotation.pdf}}
    \caption{Rotating the drawing to obtain an MW-$\beta$ parallelogram drawing.}
    \label{fig:rotation}
\end{figure}

To complete the proof, we thus show how to rotate $P$ to produce an MW-$\beta$ parallelogram drawing. Refer to \cref{fig:rotation-before}.
Let $\gamma_0$ be the angle between $\overline{r_{0,0}r_1}$ and $\overline{r_{0,0}b_{0,0}}$ and let $\gamma_1$ be the angle between $\overline{r_{1,m}r_0}$ and $\overline{r_{1,m}b_{1,m}}$;  Let $\gamma = \min\{\gamma_0, \gamma_1\}$.
Let $\hat{f_0}$ be the ray originating at $r_{0,0}$, forming an angle $\gamma' < \gamma$ with, and lying above, segment $\overline{r_{0,0}r_1}$, so that no edge of the drawing is perpendicular to $\hat{f_0}$. Let $\hat{f_1}$ be the ray originating at $r_{1,m}$ having opposite direction to $\hat{f_0}$. Observe that $\hat{f_0}$ and $\hat{f_1}$ are parallel and that any vertex of $T_i$ except $r_i$ is in the strip between $\hat{f_0}$ and $\hat{f_1}$. We now rotate $P$ counterclockwise until $\hat{f_0}$ and $\hat{f_1}$ become horizontal; see \cref{fig:rotation-after}. This produces a parallelogram drawing of $\langle T_0, T_1 \rangle$, since no edge is vertical, $y(r_0) > y(r_{1,m}) > y(r_{0,0}) > y(r_1)$, and $x(r_0) < x(r_{0,0}) < x(r_{1,m}) < x(r_1)$.
\end{proof}

\section{Pruning Leaves from MW-$\beta$ Drawings of Isomorphic Trees}\label{se:non-iso-trees}

In this section, we explore the question of how far from isomorphic two trees might be while still allowing an MW-$\beta$ drawing. We consider the MW-$\beta$ drawing $\langle \Gamma_0, \Gamma_1 \rangle$ constructed in the proof of \cref{th:main} and ask whether it is possible to prune some leaves from $\Gamma_1$ and still have an MW-$\beta$ drawing of the resulting trees. Precisely, we show that there are cases when we can remove linearly many leaves from $\Gamma_1$ and still obtain an MW-$\beta$ drawing of the resulting tree for any $\beta \in [1, \infty]$. 
It may be worth recalling that Lenhart and Liotta proved that two stars admit an MW-$1$ drawing if and only if the cardinalities of their vertex sets differ by at most two~\cite{DBLP:conf/gd/LenhartL22}.

Let $(T,r)$ be a rooted tree and let $\mathcal{L}$ be a set of leaves of $T$. The vertex $v$ is a \emph{cousin} of a vertex $v'$ if $v$ and $v'$ have a common grandparent but no common parent, i.e., there is a vertex $w$ such that a length-2 directed path $w,p,v$ and a length-2 directed path $w,p',v'$ with $p\neq p'$ exist.
We say that $\mathcal{L}\neq\emptyset$ is \emph{sparse} if, for every $v\in \mathcal{L}$,
\begin{enumerate*}[label=(\roman*)]
    \item $v$ has at least one sibling,
    \item every sibling $v'$ of $v$ is a leaf with $v'\notin\mathcal{L}$, and
    \item $v$ has a cousin $w$ such that $w\notin\mathcal{L}$ and, for all siblings $w'$ of $w$, $w'\notin\mathcal{L}$.
\end{enumerate*}
Note that the existence of a sparse set implies that $(T,r)$ has height at least 2, otherwise there is no vertex that has a cousin.

\begin{restatable}[\restateref{th:non-isomorphic}]{theorem}{NonIsomorphic}
\label{th:non-isomorphic}
Let $(T,r)$ be a rooted tree and let $\mathcal{L}$ be a sparse set of leaves of $T$. Then the pair $\langle T, T \setminus \mathcal{L} \rangle$ of trees admits an MW-$\beta$ drawing for all $\beta \in [1,\infty]$.
\end{restatable}

\begin{restatable}[\restateref{co:non-isomorphic}]{corollary}{NonIsomorphicCor}
\label{co:non-isomorphic}
For any $m\ge 1$ and $n = 7m+1$, there exist tree pairs $\langle T_0, T_1 \rangle$ with $|V(T_1)| \leq 1+\frac{5}{6}(|V(T_0)|-1)$ that admit an MW-$\beta$ drawing for all $\beta\in [1,\infty]$.
\end{restatable}

\section{Concluding Remarks}\label{se:open}
In this paper, we studied the mutual witness proximity drawability of pairs of isomorphic trees. We adopted the well-known concept of open/closed $\beta$-proximity regions and considered any value of the parameter $\beta$ such that $\beta \geq 1$. For the special case of $\beta=1$, the definition of closed $\beta$-proximity region coincides with the definition of Gabriel proximity region. We showed in~ \cref{th:caterpillar} that any pair of isomorphic caterpillars admits a linearly separable mutual witness Gabriel drawing. 
We then extended this result in~\cref{th:main} to any value of
$\beta \geq 1$ and to any pair of isomorphic trees, but at the cost of losing linear separability.  

It would be interesting to establish whether any two isomorphic trees admit a linearly separable MW-$\beta$ drawing for $\beta \geq 1$. Also, even for the special case of caterpillars, extending the result of \cref{th:caterpillar} to values of $\beta > 1$ does not seem immediate. Finally, a characterization of those non-isomorphic pairs of trees that admit a mutual witness $\beta$-drawing continues to be elusive. \cref{th:non-isomorphic} shows that the trees in the pair may differ by linearly many vertices.

\section{Acknowledgements}
 We thank Stefan Näher for many helpful discussions, for implementing the caterpillar algorithm, and for creating a program to edit and verify MW-$[1]$ drawings that was very helpful in verifying our constructions.

\clearpage
\bibliographystyle{splncs04}
\bibliography{abbrv,biblio}

\begin{thebibliography}{10}
\providecommand{\url}[1]{\texttt{#1}}
\providecommand{\urlprefix}{URL }
\providecommand{\doi}[1]{https://doi.org/#1}

\bibitem{DBLP:journals/comgeo/AronovDH11}
Aronov, B., Dulieu, M., Hurtado, F.: Witness {(Delaunay)} graphs. Comput. Geom.
   \textbf{44}(6-7),  329--344 (2011). \doi{10.1016/j.comgeo.2011.01.001}

\bibitem{DBLP:journals/comgeo/AronovDH13}
Aronov, B., Dulieu, M., Hurtado, F.: Witness {Gabriel} graphs. Comput. Geom.
  \textbf{46}(7),  894--908 (2013). \doi{10.1016/j.comgeo.2011.06.004}

\bibitem{DBLP:journals/ipl/AronovDH14}
Aronov, B., Dulieu, M., Hurtado, F.: Mutual witness proximity graphs. Inf.
  Process. Lett.  \textbf{114}(10),  519--523 (2014).
  \doi{10.1016/j.ipl.2014.04.001}

\bibitem{DBLP:journals/gc/AronovDH14}
Aronov, B., Dulieu, M., Hurtado, F.: Witness rectangle graphs. Graphs Comb.
  \textbf{30}(4),  827--846 (2014). \doi{10.1007/s00373-013-1316-x}

\bibitem{DBLP:books/ph/BattistaETT99}
Battista, G.D., Eades, P., Tamassia, R., Tollis, I.G.: Graph Drawing:
  Algorithms for the Visualization of Graphs. Prentice-Hall (1999)

\bibitem{DBLP:journals/jda/BattistaLW06}
Battista, G.D., Liotta, G., Whitesides, S.: The strength of weak proximity. J.
  Discrete Algorithms  \textbf{4}(3),  384--400 (2006).
  \doi{10.1016/j.jda.2005.12.004}

\bibitem{DBLP:journals/jgaa/EadesH0K17}
Eades, P., Hong, S., Nguyen, A., Klein, K.: Shape-based quality metrics for
  large graph visualization. J. Graph Algorithms Appl.  \textbf{21}(1),  29--53
  (2017). \doi{10.7155/jgaa.00405}

\bibitem{gs-nsagv-69}
Gabriel, K.R., Sokal, R.R.: A new statistical approach to geographic variation
  analysis. Systematic Zoology  \textbf{18},  259--278 (1969).
  \doi{10.2307/2412323}

\bibitem{DBLP:journals/pr/IchinoS85}
Ichino, M., Sklansky, J.: The relative neighborhood graph for mixed feature
  variables. Pattern Recognit.  \textbf{18}(2),  161--167 (1985).
  \doi{10.1016/0031-3203(85)90040-8}

\bibitem{jt-rngtr-92}
Jaromczyk, J.W., Toussaint, G.T.: Relative neighborhood graphs and their
  relatives. Proc. {IEEE}  \textbf{80}(9),  1502--1517 (1992).
  \doi{10.1109/5.163414}

\bibitem{DBLP:conf/dagstuhl/1999dg}
Kaufmann, M., Wagner, D. (eds.): Drawing Graphs, Methods and Models (the book
  grow out of a Dagstuhl Seminar, April 1999), Lecture Notes in Computer
  Science, vol.~2025. Springer (2001). \doi{10.1007/3-540-44969-8}

\bibitem{KR-85}
Kirkpatrick, D.G., Radke, J.D.: A framework for computational morphology.
  Machine Intelligence and Pattern Recognition  \textbf{2},  217--248 (1985).
  \doi{10.1016/B978-0-444-87806-9.50013-X}

\bibitem{DBLP:conf/gd/LenhartL22}
Lenhart, W.J., Liotta, G.: Mutual {W}itness {G}abriel {D}rawings of {C}omplete
  {B}ipartite {G}raphs. In: Angelini, P., von Hanxleden, R. (eds.) Graph
  Drawing and Network Visualization - 30th International Symposium, {GD} 2022,
  Tokyo, Japan, September 13--16, 2022, Revised Selected Papers. Lecture Notes
  in Computer Science, vol. 13764, pp. 25--39. Springer (2022).
  \doi{10.1007/978-3-031-22203-0\_3}

\bibitem{DBLP:reference/crc/Liotta13}
Liotta, G.: Proximity drawings. In: Tamassia, R. (ed.) Handbook on Graph
  Drawing and Visualization, pp. 115--154. Chapman and Hall/CRC (2013),
  \url{https://cs.brown.edu/people/rtamassi/gdhandbook/chapters/proximity.pdf}

\bibitem{DBLP:books/wi/OkabeBSCK00}
Okabe, A., Boots, B., Sugihara, K., Chiu, S.N., Kendall, D.G.: Spatial
  Tessellations: Concepts and Applications of Voronoi Diagrams, Second Edition.
  Wiley Series in Probability and Mathematical Statistics, Wiley (2000).
  \doi{10.1002/9780470317013}

\bibitem{DBLP:reference/cg/ORourkeT17}
O'Rourke, J., Toussaint, G.T.: Pattern recognition. In: Goodman, J.E.,
  O'Rourke, J., Toth, C. (eds.) Handbook of Discrete and Computational
  Geometry, Third Edition. Chapman and Hall/CRC (2017),
  \url{http://www.csun.edu/~ctoth/Handbook/chap54.pdf}

\bibitem{DBLP:reference/crc/2013gd}
Tamassia, R. (ed.): Handbook on Graph Drawing and Visualization. Chapman and
  Hall/CRC (2013),
  \url{https://www.crcpress.com/Handbook-of-Graph-Drawing-and-Visualization/Tamassia/9781584884125}

\bibitem{DBLP:reference/cg/TamassiaL04}
Tamassia, R., Liotta, G.: Graph drawing. In: Goodman, J.E., O'Rourke, J. (eds.)
  Handbook of Discrete and Computational Geometry, Second Edition, pp.
  1163--1185. Chapman and Hall/CRC (2004). \doi{10.1201/9781420035315.ch52}

\bibitem{DBLP:journals/pr/Toussaint80}
Toussaint, G.T.: The relative neighbourhood graph of a finite planar set.
  Pattern Recognit.  \textbf{12}(4),  261--268 (1980).
  \doi{10.1016/0031-3203(80)90066-7}

\bibitem{DBLP:conf/mldm/ToussaintB12}
Toussaint, G.T., Berzan, C.: Proximity-graph instance-based learning, support
  vector machines, and high dimensionality: An empirical comparison. In:
  Perner, P. (ed.) Machine Learning and Data Mining in Pattern Recognition -
  8th International Conference, {MLDM} 2012, Berlin, Germany, July 13-20, 2012.
  Proceedings. Lecture Notes in Computer Science, vol.~7376, pp. 222--236.
  Springer (2012). \doi{10.1007/978-3-642-31537-4\_18}

\end{thebibliography}

\clearpage

\appendix

\section{Omitted Proofs from \cref{se:iso-caterpillars}}

\LemStars*
\label{le:stars*}

\begin{proof}
We first consider the case $k=0$; see \cref{fi:stars-zero-app}. We place $r_0$ at position $a_0=(0,5)$, $v_{0,0}$ at position $b_0=(0,3)$, $r_1$ at position $a_1=(2,0)$, and $v_{1,0}$ at position $b_1=(2,2)$. This way, the angle inside $P=(a_0,b_0,a_1,b_1)$ at $a_0$ is smaller than $\pi/4$. We place the anchor $q_0$ at $(1.1,3)$ and $q_1$ at $(0.9,2)$ and compute the safe wedges $W_0,W_1$ and ports $p_0,p_1$ as above to obtain an MW-$[1]$-drawing inside a winged parallelogram $WP(P,q_0,q_1,W_0,W_1,p_0,p_1)$ with the desired properties.

Consider now the case that $k>0$. We will create an MW-$[1]$ drawing $\langle \Gamma_0,\Gamma_1\rangle$ that is point symmetric in the origin, i.e., a drawing with $x(r_1)=-x(r_0)$, $y(r_1)=-y(r_0)$, and $x(v_{1,i})=-x(v_{0,k-i+1})$, $y(v_{1,i})=-y(v_{0,k-i+1})$ for $1\le 0\le k$. By symmetry, we only have to argue that the edges of $T_0$ are realized correctly.

We first place the leaves $v_{0,i}$, $0\le i\le k$, at y-coordinate $0.5$ and x-coordinate $2i-k+0.5$, such that any pair $v_{0,i},v_{0,i+1}$ has distance 2 and $\angle(v_{0,i},v_{1,i},v_{0,i+1})=\pi/2$; see \cref{fi:stars-leaves-app}. Thus, $v_{1,i}$ lies in $R[v_{0,i},v_{0,j},1]$ for any two vertices with $0\le i<j\le k$, so no two leaves are adjacent.

Now we place $r_0$ with $x(r_0)=x(v_{0,0})$; see \cref{fi:stars-root-app}. We have to make sure that the regions $R[r_0,v_{0,i},1]$ 
contains no witness. Observe that, by definition, any Gabriel region $R[u,v,1]$ is contained in the disk around $u$ with radius 
$d(u,v)$. Hence, no point $w$ with $d(u,w)>d(u,v)$ can lie in $R[u,v,1]$. Consider the point 
$p=(x(v_{0,0}),y(v_{1,0}))=(2i-k+0.5,-0.5)$. By construction, we have $d(r_0,v_{1,i})<d(r_0,p)=y(r_0)+0.5$ and we want to make sure 
that $d(r_0,v_{0,i})\le d(r_0,v_{0,k}) < d(r_0,p)$ for each $0\le i\le k$. Consider now the triangle $\triangle(r_0,p,v_{0,k})$. If 
$\angle(p,v_{0,k},r_0)=\pi/2$, then by Pythagoras $d(r_0,p)>d(r_0,v_{0,k})$. Let $\alpha=\angle(v_{0,k},r_0,p)$ and 
$\beta=\angle(r_0,p,v_{0,k})$ with $\alpha+\beta=\pi/2$. Consider the point $q=(x(v_{0,k}),y(v_{1,k})=(k-0.5,-0.5)$ and consider the 
triangle $\triangle(p,q,v_{0,k})$. We have $\angle(r_0,p,q)=\pi/2$ and $\angle(r_0,p,v_{0,k})=\beta$, so $\angle(v_{0,k},p,q)=\alpha$. 
Since $\angle(p,q,v_{0,k})$, we also have $\angle(q,v_{0,k},p)=\beta$. Hence, the triangles $\triangle(r_0,p,v_{0,k})$ and 
$\triangle(p,q,v_{0,k})$ are congruent, so we have $\frac{d(r_0,p)}{d(p,v_{0,k})} = \frac{d(p,v_{0,k})}{d(v_{0,k},q)}$. Since 
$d(v_{0,k},q)=1$ by choice of $q$, we thus have $d(r_0,p)=d(p,v_{0,k})^2$. By Pythagoras' Theorem, 
$d(p,v_{0,k})^2 = d(p,q)^2+d(v_{0,k},q)^2 = 2k^2+1$. Thus, $y(r_0)=d(r_0,p)-0.5 = 2k^2+0.5$ ensures that no edges of $T_0$ has a witness in $\Gamma_1$. Furthermore, note that, by construction, $\beta$ is larger than $\pi/4$ since $1=d(v_{0,0},p)<d(v_{0,0},v_{0,k})$ as 
long as $k>0$, so $\alpha$ is smaller than $\pi/4$.

We choose the winged parallelogram $WP((a_0,b_0,a_1,b_1), q_0, q_1, W_0, W_1,p_0,p_1)$ as follows. For $(a_0,b_0,a_1,b_1)$, we choose the positions of $r_0,v_{0,0},r_1,v_{1,k}$, respectively. We place the point $q_0$ slightly to the right of $v_{0,k}$ at $(x(v_{0,k})+\varepsilon,0.5)$, and the point $q_1$ slightly to the left of $v_{1,0}$ at $(x(v_{1,0}-\varepsilon,-0.5)$ for some small enough $\varepsilon$.
The interior angle at points $a_0,a_1$ is smaller than $\pi/4$ as long as $2k^2+1>2k+1$, which is true for $k>0$,  and all leaves are placed on the desired positions. We choose the safe wedges $W_0,W_1$ and ports $p_0,p_1$ as in the definition. For an illustration, see \cref{fi:stars-drawing-app}.
\end{proof}

\begin{figure}[!t]
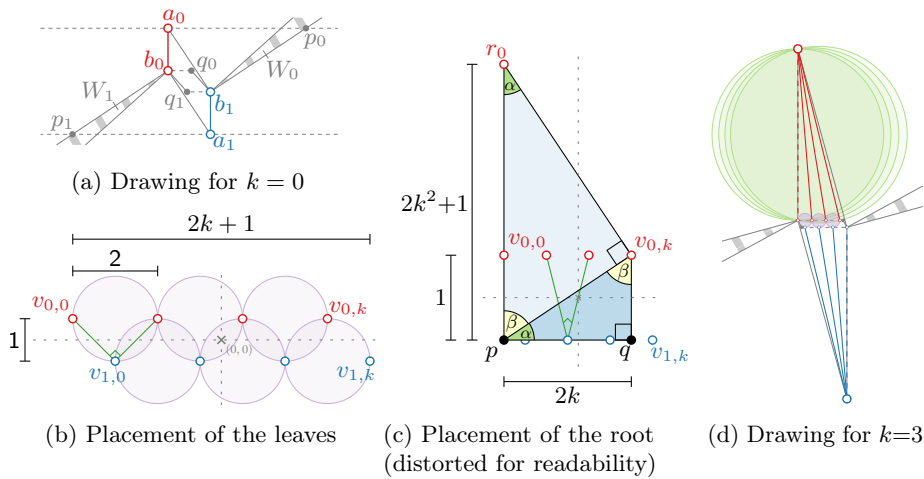

	\centering
 
    \begin{minipage}[b][5.3cm]{.4\columnwidth}
	 \subcaptionbox{Drawing for $k=0$\label{fi:stars-zero-app}}[\linewidth]{\includegraphics[page=5]{images/stars-construction.pdf}}
      \vfill
	 \subcaptionbox{Placement of the leaves\label{fi:stars-leaves-app}}{\includegraphics[page=1]{images/stars-construction.pdf}}
    \end{minipage}
    \hfill
    \subcaptionbox{Placement of the root\\ (distorted for readability)\label{fi:stars-root-app}}{~\includegraphics[page=2]{images/stars-construction.pdf}~}
    \hfill
    \subcaptionbox{Drawing for $k{=}3$\label{fi:stars-drawing-app}}{\includegraphics[page=4]{images/stars-construction.pdf}}
	\caption{Illustration for the Proof of \cref{le:stars}.}
	\label{fi:stars-construction-app}
\end{figure}

\section{Omitted Proofs from \cref{se:non-iso-trees}}

\NonIsomorphic*
\label{th:non-isomorphic*}

We start by a definition and a technical lemma. Let $\langle \Gamma_0, \Gamma_1 \rangle$ be a MW-$\beta$ parallelogram drawing of two trees in a parallelogram $P=\langle a_0,b_0,a_1,b_1\rangle$; see \cref{fig:strip-ratio-def}. The \emph{strip ratio} $\sigma(\Gamma_0,\Gamma_1)$ of $\langle \Gamma_0, \Gamma_1 \rangle$ is defined as \[\sigma(\Gamma_0,\Gamma_1)=\frac{|y(b_1)-y(b_0)|}{|y(a_0)-y(a_1)|}.\]

\begin{figure}[t]
    \centering
    \subcaptionbox{The strip ratio\label{fig:strip-ratio-def}}{\includegraphics[page=1]{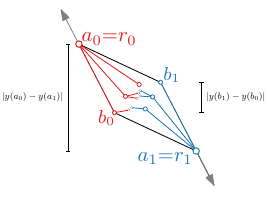}}
    \hfil
    \subcaptionbox{Lowering the strip ratio\label{fig:strip-ratio-lower}}{\includegraphics[page=2]{strip-ratio}}
    \caption{Illustration of the strip ratio and the proof of \cref{le:strip-ratio}.}
    \label{fig:strip-ratio}
\end{figure}

\begin{lemma}\label{le:strip-ratio}
Let $\langle T_0, T_1 \rangle$ be two isomorphic trees and let $\varepsilon > 0$ be an arbitrarily small real number. There exists a parallelogram drawing of $\langle T_0, T_1 \rangle$ whose strip ratio is $\sigma < \varepsilon$.
\end{lemma}
\begin{proof}
    We construct a MW-$\beta$ parallelogram drawing $\langle \Gamma_0,\Gamma_1\rangle$ for $\langle T_0,T_1\rangle$. If $\sigma(\Gamma_0,\Gamma_1)<\varepsilon$, then we are done. 
    Otherwise, we simultaneously move $a_0$ (and thus $r_0$) along the ray $b_0a_0$ upwards and $a_1$ (and thus $r_1$) along the ray $b_1a_1$ downwards until $\sigma(\Gamma_0,\Gamma_1)<\varepsilon$; see \cref{fig:strip-ratio}. Note that this movement corresponds to moving $a_0$ ($r_0$) vertically upwards and $a_1$ ($r_1$) vertically downwards before the final rotation step in the proof of \cref{th:main}. Since, for the proof of correctness, it was only important that these two points are far enough above/below the other vertices, the drawing remains an MW-$\beta$ drawing.
\end{proof}

\begin{proof}[of \cref{th:non-isomorphic}]
    First, note that, for any subtree $(T',r')$ of $(T,r)$ rooted in $r'$ of height at least 2, $\mathcal{L}\cap V(T')$ is sparse for $\langle T',r'\rangle$.

    We show by induction on the height $\delta\ge 2$ of $(T,r)$ that an MW-$\beta$ drawing can always be produced. 

    Consider first the base case $\delta = 2$. By definition of sparse sets, the children of $r$ cannot be in $\mathcal{L}$, as they have no cousins. Let $(T_0,r_0)=(T,r)$ and let $(T_{0,0}, r_{0,0}), \ldots, (T_{0,m}, r_{0,m})$ be the subtrees of $(T_0,r_0)$ resulting from deleting $r_0$ from $T_0$. Then each tree $(T_{0,j},r_{0,j})$ is of one of three types; see \cref{fig:almost-isomorphic-base}:
    \begin{enumerate}[label=(\Alph*)]
        \item $r_{0,j}$ is a leaf not in $\mathcal{L}$,
        \item $(T_{0,j},r_{0,j})$ has height 1 with exactly one of its leaves  $v_{0,j}\in\mathcal{L}$,
        \item $(T_{0,j},r_{0,j})$ has height 1 with no leaf in $\mathcal{L}$.
    \end{enumerate}
    Note that there must be at least one subtree of type (C), but there may be no subtrees of type (A) or (B). We now reorder the children of $r_0$ such that, from left to right, we first have all subtrees of type (A), then all subtrees of type (B), and then all subtrees of type (C). Within each subtree of $(T_{0,j},r_{0,j})$ type (B), we order the leaves such $v_{0,j}$ is the rightmost leaf; see \cref{fig:almost-isomorphic-base}.

    \begin{figure}[t]
        \centering
        \includegraphics[page=1]{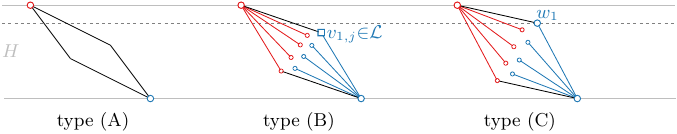}
        \caption{The three types of subtrees in the base case of \cref{th:non-isomorphic}.}
        \label{fig:almost-isomorphic-base}
    \end{figure}

    Let $(T_1,r_1)$ be isomorphic to $(T_0,r_0)$. We first compute a MW-$\beta$ parallelogram drawing of $\langle (T_0,r_0),(T_1,r_1)\rangle$ in a parallelogram $P=(a_0,b_0,a_1,b_0)$ according to the proof of \cref{th:main}, but with some small adjustments. Using \cref{le:strip-ratio}, we ensure that the rightmost subtree $(T_{i,m},r_{i,m})$, $0\le i\le 1$, which is of type (C), has the largest strip ratio among all subtrees $(T_{i,j},r_{i,j})$. Let $w_1$ be the rightmost leaf of $(T_{1,m},r_{1,m})$. Then, placing the subtrees $(T_{i,j},r_{i,j})$ in the horizontal strip $H$ as in the proof of \cref{th:main}, $w_1$ will be the rightmost and topmost vertex of $T_1$ in the interior of $H$.
    
    We place $r_0$ and $r_1$ as in the proof of \cref{th:main}, but with the additional constraint that for every vertex $u_0$ of $(T_0,r_0)$ in the interior of $H$, $\angle(u_0,w_1,r_0)\ge\pi/2$, so that $w_1$ lies in the $\beta$-region $R[r_0,u_0,\beta]$. Similar to the proof of \cref{le:strip-ratio}, this can be achieved by moving $r_0$ upwards along the ray $b_0r_0$. 
    Since $w_1$ belongs to a subtree of type (C), $w_1\notin\mathcal{L}$. Hence, after removing the leaves of $\mathcal{L}$, all edges between $r_0$ and any non-adjacent vertex $u_0$ of $T_0$ (which lies in the interior of $H$) still have $w_1$ as a witness; see \cref{fig:almost-isomorphic-base-rotation}.

    \begin{figure}[t]
        \centering
        \includegraphics[page=2]{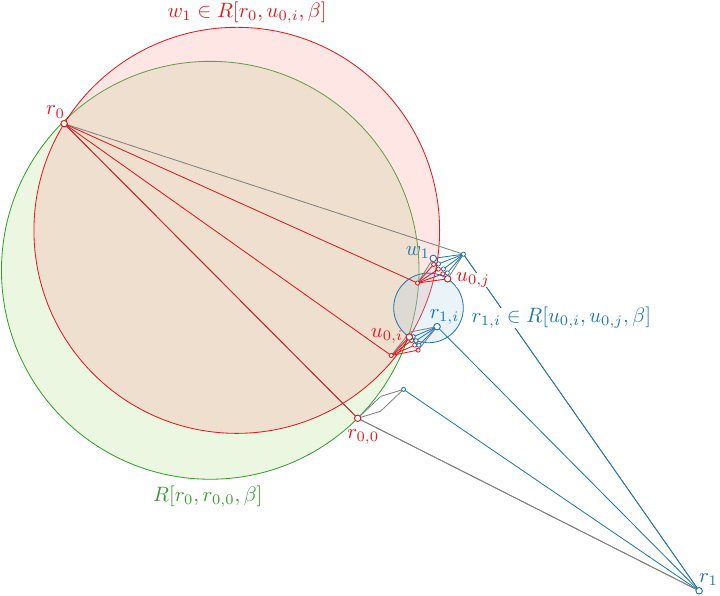}
        \caption{The drawing in the base case of \cref{th:non-isomorphic} after removing the leaves of $\mathcal L$. There is not witness in $R[r_0,r_{0,0},\beta]$ (green disk); $w_1\in R[r_0,u_{0,i},\beta]$ (red disk); $r_{1,i}\in R[u_{0,i},u_{0,j},\beta]$ (blue disk)}
        \label{fig:almost-isomorphic-base-rotation}
    \end{figure}
    
    Note that $r_{1,m}$ is placed at point $b_1$ of the parallelogram, and since $(T_{1,m},r_{1,m})$ is of type (C), $r_{1,m}$ is not a leaf and thus $r_{1,m}\notin\mathcal{L}$, so removing the leaves of $\mathcal{L}$ from $\Gamma_1$ does not destroy the MW-$\beta$ parallelogram drawing properties. 
    
    Furthermore, for any two vertices $u_{0,i}$ in $(T_{0,i},r_{0,i})$ and $u_{0,j}$ in $(T_{0,j},r_{0,j})$ with $0\le i<j\le m$ where $(T_{0,i},r_{0,i})$ is of type (B), we have that $r_{1,i}$ lies in the $\beta$-region $R[u_{0,i},u_{0,j},\beta]$, so we can remove $v_{1,i}$ from any subtree of type (B) without destroying the MW-$\beta$ drawing properties. Hence, we obtain a parallelogram MW-$\beta$ drawing of $\langle T, T\setminus\mathcal{L}\rangle$. Note that the strip ratio of the drawing can also be lowered by moving $r_0$ upwards along the ray $b_0r_0$ and $r_1$ downwards along the ray $b_1r_1$ as in the proof of \cref{le:strip-ratio}.

    Consider now the inductive case of $\delta>2$. Let $(T_0,r_0)=(T,r)$ and let $(T_{0,0}, r_{0,0}), \ldots, (T_{0,m}, r_{0,m})$ be the subtrees of $(T_0,r_0)$ resulting from deleting $r_0$ from $T_0$. Then each tree $(T_{0,j},r_{0,j})$ is of one of four types; see \cref{fig:almost-isomorphic-induction}:
    \begin{enumerate}[label=(\Alph*)]
        \item $r_{0,j}$ is a leaf not in $\mathcal{L}$,
        \item $(T_{0,j},r_{0,j})$ has height 1 with exactly one of its leaves  $v_{0,j}\in\mathcal{L}$,
        \item $(T_{0,j},r_{0,j})$ has height 1 with no leaf in $\mathcal{L}$,
        \item $(T_{0,j},r_{0,j})$ has height at least 2 but smaller than $\delta$.
    \end{enumerate}
    Note that there must be at least one subtree of type (D), but there might be no subtrees of type (A), (B) or (C). We now reorder the children of $r_0$ such that, from left to right, we first have all subtrees of type (A), then all subtrees of type (B), then all subtrees of type (C), and then all subtrees of type (D). Within each subtree of $(T_{0,j},r_{0,j})$ type (B), we order the leaves such $v_{0,j}$ is the rightmost leaf; see \cref{fig:almost-isomorphic-induction}. 

    \begin{figure}[t]
        \centering
        \includegraphics[page=3]{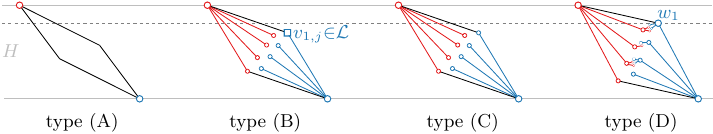}
        \caption{The four types of subtrees in the induction step of \cref{th:non-isomorphic}.}
        \label{fig:almost-isomorphic-induction}
    \end{figure}

    Let $\mathcal{L}'\subseteq \mathcal{L}$ be the set of leaves of $\mathcal{L}$ in the subtrees of type (D). Let $(T_1,r_1)$ be isomorphic to $(T_0\setminus\mathcal{L}',r_0)$. By induction, every pair of subtrees $\langle (T_{0,j},r_{0,j}),(T_{1,j},r_{1,j})\rangle$ of type (D) has a parallelogram MW-$\beta$ drawing where the strip ratio can be arbitrarily lowered.
    
    We arrange the parallelogram drawings of the subtrees $\langle (T_{0,j},r_{0,j}),(T_{1,j},r_{1,j})\rangle$ inside a horizontal strip $H$ as in the base case, using \cref{le:strip-ratio} to ensure that the drawing of $\langle (T_{0,m},r_{0,m}),(T_{1,m},r_{1,m})\rangle$, which is of type (D), has the largest strip ratio among all pairs of subtrees $\langle (T_{0,m},r_{0,m}),(T_{1,m},r_{1,m})\rangle$; see \cref{fig:almost-isomorphic-induction}. 
    
    Let $w_1$ be the topmost (and rightmost) vertex of $(T_{1,m},r_{1,m})$ inside $H$. We again move $r_0$ upwards along the ray $b_0r_0$ such that, For every vertex $u_0$ of $T_0$ in the interior of $H$, $\angle(u_0,w_1,r)\ge\pi/2$, so that $w_1$ lies in the $\beta$-region $R[r_0,u_0,\beta]$; see \cref{fig:almost-isomorphic-induction-rotation}. Since $w_1\notin\mathcal{L}$ (otherwise it would not be in $(T_{1,m},r_{1,m})$, as we already removed the leaves of $\mathcal{L'}$), after removing the leaves of $\mathcal{L}$, all edges between $r_0$ and any non-adjacent vertex $u_0$ of $T_0$ (which lies in the interior of $H$)  still have $w_1$ as a witness. Furthermore, the edges between disjoint subtrees still have witnesses following the same argument as in the base case, and $r_{1,m}$, which is not a leaf and thus not in $\mathcal{L}$, lies at point $b$ of the parallelogram. Hence, after removing all leaves of $\mathcal{L}$, we obtain a MW-$\beta$ parallelogram drawing of $\langle T,T\setminus \mathcal{L}\rangle$.
\end{proof}

\NonIsomorphicCor*
\label{co:non-isomorphic*}

\begin{proof}
    We construct an infinite family of trees and sets of leaves as follows. For any $m>0$, $(T,r)$ is a tree rooted in $r$ such that removing $r$ yields $m$ subtrees $(T_{0},r_{0}),\ldots,(T_{m-1},r_{m-1})$. Every subtree $T_{j}$, $0\le j<m$ consists of the following; see \cref{fig:corollary}.
    \begin{enumerate}[label=(\roman*)]
        \item The root $r_{j}$ has 2 children $u_{j}$ and $u'_{j}$;
        \item $u_{j}$ has one child $v_{j}$ which is a leaf
        \item $u'_{j}$ has two children $w_{j}$ and $w'_{j}$ which are leaves with $w'_{j}\in\mathcal{L}$.
    \end{enumerate}
    Then $\mathcal{L}$ is sparse, so $\langle T,T\setminus\mathcal{L}\rangle$ admits an MW-$\beta$ drawing by \cref{th:main}. Every subtree $T_j$ has 6 vertices, so $|V(T)|=6m+1$. $\mathcal{L}$ has one leaf per subtrees $T_j$, so $|\mathcal{L}|=m$ and thus $|V(T\setminus \mathcal{L})|=5m+1=1+\frac{5}{6}(|V(T)|-1)$.
\end{proof}

\begin{figure}[t]
    \centering
    \includegraphics[page=4]{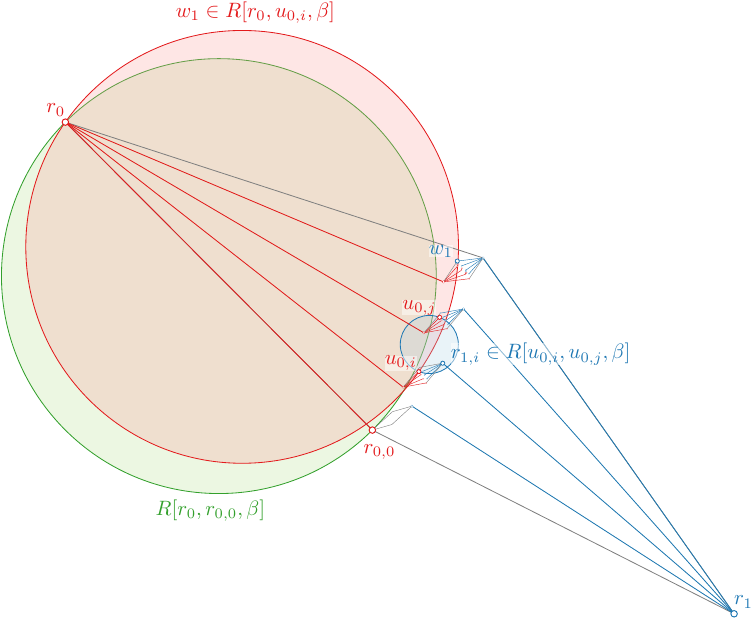}
    \caption{The drawing in the induction step of \cref{th:non-isomorphic} after removing the leaves of~$\mathcal L$. There is not witness in $R[r_0,r_{0,0},\beta]$ (green disk); $w_1\in R[r_0,u_{0,i},\beta]$ (red disk); $r_{1,i}\in R[u_{0,i},u_{0,j},\beta]$ (blue disk)}
    \label{fig:almost-isomorphic-induction-rotation}
\end{figure}

\begin{figure}[b]
    \centering
    \includegraphics{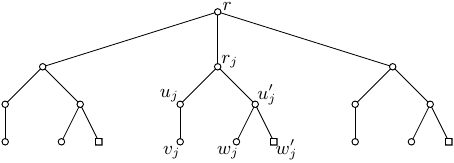}
    \caption{Construction of $(T,r)$ in the proof of \cref{co:non-isomorphic}.}
    \label{fig:corollary}
\end{figure}

\end{document}